\documentclass[a4paper,UKenglish,cleveref,thm-restate]{lipics-v2021}

\usepackage[subrefformat=simple,labelformat=simple,position=b]{subcaption}

\usepackage{graphicx}
\usepackage{amssymb}
\usepackage{amsmath}
\usepackage{complexity}
\usepackage{thm-restate}
\usepackage{cleveref}
\usepackage{xcolor}
\usepackage{xspace}
\usepackage{nicefrac}
\usepackage{enumerate}
\usepackage{siunitx}
\usepackage{bm}

\definecolor{darkgreen}{rgb}{0.164706, 0.384314, 0.0941176}

\newcommand{\Up}{\ensuremath{u}\xspace}
\newcommand{\Down}{\ensuremath{d}\xspace}
\newcommand{\Right}{\ensuremath{r}\xspace}
\newcommand{\Left}{\ensuremath{l}\xspace}

\newcommand{\decProb}{\textsc{Min-Gathering}\xspace}

\newcommand{\Part}{\ensuremath{\mathcal{P}}\xspace}
\newcommand{\diam}{\ensuremath{D}\xspace}
\DeclareMathOperator{\dist}{dist}

\hideLIPIcs

\graphicspath{{./figures/}}

\title{Targeted Drug Delivery: Algorithmic Methods for Collecting a Swarm of Particles with Uniform External Forces}

\titlerunning{Collecting a Swarm of Particles with Uniform External Forces}

\author{Aaron T. Becker}{Department of Electrical and Computer Engineering, University of Houston, USA}{atbecker@uh.edu}{https://orcid.org/0000-0001-7614-6282}{}
\author{Sándor P. Fekete}{Department of Computer Science, TU Braunschweig, Braunschweig, Germany}{s.fekete@tu-bs.de}{https://orcid.org/0000-0002-9062-4241}{}
\author{Li Huang}{Department of Electrical and Computer Engineering, University of Houston, USA}{lhuang28@uh.edu}{https://orcid.org/0000-0001-9559-7724}{}
\author{Phillip Keldenich}{Department of Computer Science, TU Braunschweig, Braunschweig, Germany}{keldenich@ibr.cs.tu-bs.de}{https://orcid.org/0000-0002-6677-5090}{}
\author{Linda Kleist}{Department of Computer Science, TU Braunschweig, Braunschweig, Germany}{kleist@ibr.cs.tu-bs.de}{https://orcid.org/0000-0002-3786-916X}{}
\author{Dominik Krupke}{Department of Computer Science, TU Braunschweig, Braunschweig, Germany}{krupke@ibr.cs.tu-bs.de}{https://orcid.org/0000-0003-1573-3496}{}
\author{Christian Rieck}{Department of Computer Science, TU Braunschweig, Braunschweig, Germany}{rieck@ibr.cs.tu-bs.de}{https://orcid.org/0000-0003-0846-5163}{}
\author{Arne Schmidt}{Department of Computer Science, TU Braunschweig, Braunschweig, Germany}{aschmidt@ibr.cs.tu-bs.de}{https://orcid.org/0000-0001-8950-3963}{}

\authorrunning{A. T. Becker et al.}

\Copyright{Aaron T. Becker, Sándor P. Fekete, Li Huang, Phillip Keldenich, Linda Kleist, Dominik Krupke, Christian Rieck, Arne Schmidt}

\ccsdesc{Theory of computation~Computational geometry}

\keywords{TILT, \NP-completeness, particles, gathering, reinforcement learning}

\relatedversion{Full version of a paper that appeared in the proceedings of ICRA 2020~\cite{BeckerFHKKKR020}.}

\nolinenumbers

\EventEditors{}
\EventNoEds{1}
\EventLongTitle{}
\EventShortTitle{}
\EventAcronym{}
\EventYear{}
\EventDate{}
\EventLocation{}
\EventLogo{}
\SeriesVolume{}
\ArticleNo{}

\begin{document}

    \maketitle

    \begin{abstract}
		We investigate algorithmic approaches for targeted drug delivery in a complex, maze-like environment, such as a vascular system. 
		The basic scenario is given by a large swarm of micro-scale particles (``agents'') and a particular target region (``tumor'') within a system of passageways.
		Agents are too small to contain on-board power or computation and are instead controlled by a global external force that acts uniformly on all particles, such as an applied fluidic flow or electromagnetic field. 
		The~challenge is to deliver all agents to the target region with a minimum number of actuation steps.
		
		We provide a number of results for this challenge. 
		We show that the underlying problem is \NP-complete, which explains why previous work did not provide provably efficient algorithms.
		We~also develop several  algorithmic approaches that greatly improve the worst-case guarantees for the number of required actuation steps.
		We evaluate our algorithmic approaches by numerous simulations, both for deterministic algorithms and searches supported by deep learning,	which show that the performance is practically promising.
    \end{abstract}


\section{Introduction}\label{section:tilt-gathering}
A crucial challenge for a wide range of vital medical problems,
such as the treatment of cancer, localized infections and inflammation,
or internal bleeding is to deliver active substances to a specific
location in an organism. The traditional approach of administering
a sufficiently large supply of these substances into the circulating blood 
may cause serious side effects,
as the outcome intended for the target site may also occur
in other places, with often undesired, serious consequences. 

Moreover, novel custom-made substances that are specifically designed
for precise effects are usually in too short supply 
to be generously poured into the blood stream. In the 
context of targeting brain tumors, as illustrated in~\cref{fig:brains}, an 
additional difficulty is the blood-brain barrier. 
This makes it necessary to develop other, more focused methods for delivering substances to specific target regions. 
\begin{figure}[ht]
	\centering
	\includegraphics[scale=0.31]{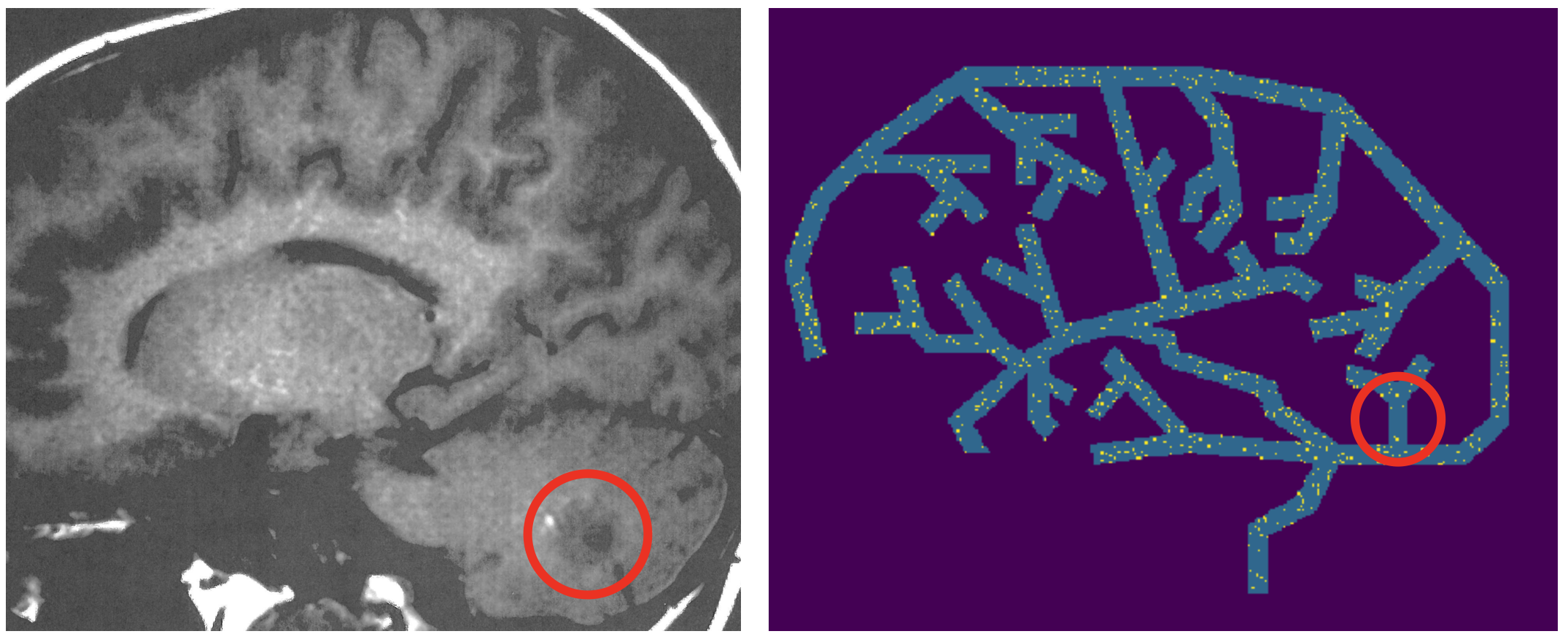}\hfill
	\caption{(Left) An MRI image of a brain tumor (marked by the red circle), located in the cerebellum.
		(Right) How can the particles (indicated by yellow dots) be delivered to the target~region?}
	\label{fig:brains}
\end{figure}

Given the main scenario of medical applications, this requires dealing with navigation through complex vascular systems, in which access to a target location is provided by pathways (in form of blood vessels) through a maze of obstacles. 
However, the microscopic size of particles necessary for passage through these vessels makes it prohibitively difficult to store sufficient energy in suitably sized microrobots, in particular in the presence of flowing~blood. 

A promising alternative is offered by employing a global external force, for example, a fluidic flow or an electromagnetic field. 
When such a force is applied, all particles move in the same direction by the same distance, unless they are blocked by obstacles in their way. 
While this makes it possible to move all particles at once, it introduces the difficulty of using \emph{uniform} forces for many particles in \emph{different} locations with different local topology to navigate them to \emph{one} final destination. 
Previous work~\cite{gathering_swarm} described a basic approach that delivers all particles in a grid environment with $n$ grid cells to a target in at most $\mathcal{O}(n^3)$ parallel moves.
This shows that delivery can always be achieved; however, a delivery time of this magnitude is usually impractical, which is why we are interested in improvements.

In particular, the problem we consider is stated as follows.

\medskip
\noindent\textsc{\underline{Tilt Gathering Problem (Min-Gathering)}}

\smallskip
\noindent\textsc{Given:} A polyomino $P$, a set of particles within $P$, and an integer $\ell$.\\
\noindent\textsc{Task:} Decide whether there exists a gathering sequence of length at most $\ell$.
\medskip

\subsection{Related work.}
We seek to understand control for large numbers of microrobots, and use a generalized model that could apply to a variety of drug-carrying microparticles.
An example are particles with a magnetic core and a catalytic surface for carrying medicinal payloads~\cite{litvinov2012high,pouponneau2009magnetic}.
An alternative are aggregates of \emph{superparamagnetic iron oxide microparticles}, 9 $\mu$m particles that are used as a contrast agent in MRI studies~\cite{mellal2015magnetic}. Real-time MRI scanning can allow feedback control using the location of a swarm of these particles.

Steering magnetic particles using the magnetic gradient coils in an MRI
scanner was implemented in~\cite{mathieu2007magnetic, pouponneau2009magnetic}.
3D Maxwell-Helmholtz coils are often used for precise magnetic field control~\cite{mellal2015magnetic}. Still needed are motion planning algorithms to guide
the swarms of robots through vascular networks.
First techniques for controlling many simple robots with
uniform control inputs presented in~\cite{reconfiguring_swarm,computation3_swarms,computation_swarms}; 
see video and abstract~\cite{bmd+-pcdfbm-15}
for a visualizing overview.
For a recent survey on challenges related to controlling multiple microrobots (less than \num{64} robots at a time), see~\cite{Chowdhury2015}.
Other related work includes assembling shapes by global control (e.g., see~\cite{full_tilt}) or rearranging particles in a
rectangle of agents in a confined workspace~\cite{zhang2017rearranging,zhang2018assembling}.
Konitzny et al.~\cite{KonitznyLLFB22} consider reinforcement methods to tackle the problem of gathering physical particles.

As the underlying problem consists of bringing together a number of agents in 
one location, a highly relevant algorithmic line of research 
considers \emph{rendezvous search}, which requires two or more
independent, intelligent agents to meet.  Alpern and Gal~\cite{alpernbook}
introduced a wide range of models and methods for this concept as have Anderson
and Fekete~\cite{anderson2001two} in a two-dimensional geometric setting. Key
assumptions include a bounded topological environment and robots with limited
onboard computation.  This is relevant to maneuvering particles through worlds
with obstacles and implementation of strategies to reduce computational burden
while calculating distances in complex worlds~\cite{meghjani2012multi,zebrowski2007energy}. In a
setting with autonomous robots, these can move independent of each other, i.e.,
follow different movement protocols, called \emph{asymmetric} rendezvous in the
mathematical literature~\cite{alpernbook}. If the agents are required to
follow the same protocol, this is called \emph{symmetric} rendezvous. This
corresponds to our model in which particles are bound by the uniform motion
constraint; symmetry is broken only by interaction with the obstacles.
For an overview of a variety of other algorithmic results on 
gathering a swarm of autonomous robots, see the survey by Flocchini~\cite{flocchini2019distributed};
note that these results assume a high degree of autonomy and computational 
power for each individual agent, so their applicability for our scenarios is
quite limited.

\subsection{Preliminaries.}
This paper focuses on planar workspaces $P$, consisting of orthogonal sets of \emph{cells} or \emph{pixels}, that form an edge-to-edge connected domain in the integer planar grid, i.e., a \emph{polyomino}, as, e.g., illustrated in~\cref{fig:hardness}.
We consider particles without autonomy and of size that is insignificant compared to the elementary cells of~$P$.
Each pixel can be referenced by its Cartesian coordinates $\bm{x}=(x,y)$.  
A pixel in $P$ is considered \emph{free} for particles entering in the manner described below; pixels not belonging to $P$ are \emph{blocked}, i.e., they form obstacles that stop the motion of particles from an adjacent cell.
The particles are \emph{commanded} in unison: In each step, all particles are relocated by one unit in one of  the directions ``Up'' (\Up), ``Down'' (\Down), ``Left'' (\Left), or ``Right'' (\Right), unless the destination is a blocked pixel; in this case, a particle remains in its previous pixel.
Because the particles are small, many of them can be located in the same pixel.
During the course of a command sequence, two particles $\pi_1$ and $\pi_2$ may end up in the
same pixel $p$, if $\pi_1$ moves into $p$, while $\pi_2$ remains in $p$ due to a blocked pixel.
Once two particles share a pixel, any subsequent command will relocate them in unison, i.e., they will not be separated, so they can be considered to be~\emph{merged}.

A \emph{motion plan} is a command sequence $C=\langle c_1,c_2,c_3,\dots\rangle$, where each command ${c_i\in\{\Up,\Down,\Left,\Right\}}$. 
For a command sequence $C$ and a non-negative integer~$\ell$, we denote the command sequence consisting of $\ell$ repetitions of $C$ by $C^\ell$.

The distance $\dist(p,q)$ between two pixels $p$ and $q$ is the length of a shortest path on the integer grid between $p$ and $q$ that stays within $P$.
The \emph{diameter} $\diam$ of a polyomino is the maximum distance between any two of its cells. 
A \emph{configuration} of $P$ is a set of cells containing at least one particle. 
The set of all possible configurations of~$P$ is denoted~by~$\Part$.
We call a command sequence \emph{gathering} if it transforms a configuration $A\in\Part$ into a configuration $A'$ such that $|A'|=1$, i.e., if it merges all particles in the same pixel.

\subsection{Our Contribution.}
In this paper, we provide a number of insights:
\begin{itemize}
\item We prove that  minimizing the length of a command sequence for gathering all particles is \NP-complete, even for environments that consist of grid cells in the plane, so no polynomial-time algorithms can be expected. 
This explains the observed difficulty of the problem, and also implies hardness for the related localization problem.
\item We develop an algorithmic strategy for gathering all particles with a worst-case guarantee of at most $\mathcal{O}(kD^2)$ steps; here $D$ denotes the maximum distance between any two points of the environment and $k$ the number of its convex corners. 
Both $k$ and $D$ are usually much smaller than the number $n$ of grid locations in the environment: $n$ may be in $\Omega(D^2)$, for two-dimensional and in $\Omega(D^3)$ for three-dimensional environments.
\item For the special case of hole-free environments, we can gather all particles in $\mathcal{O}(kD)$ steps.
\item We apply deep learning to search for short command sequences in individual, complex instances, and perform a simulation study of various approaches, evaluating the respective performance for application-inspired instances.
\end{itemize}


\section{\NP-completeness for thin polyominoes}\label{sec:gathering-hard}
In this section, we show that the problem is \NP-complete, even for thin polyominoes. 
Our proof is a reduction from \textsc{3Sat}.
Therefore, we construct for every instance~$\varphi$ of~\textsc{3Sat} a thin polyomino; for an exemplary illustration, see~\cref{fig:hardness}.
Within the polyomino, we place one particle in the top of every clause gadget, and one particle in the top left corner of each variable block.
We then argue, that a gathering sequence of a certain length exists, if and only if the formula is satisfiable. 

\begin{theorem}\label{thm:Hard}
	\decProb is \NP-complete, even for thin polyominoes.
\end{theorem}

\begin{proof}
	Clearly, a proposed sequence can be checked in polynomial time; thus, the problem is contained in \NP. For showing \NP-hardness, we reduce from \textsc{3Sat}. 
	We refer to~\cref{fig:hardness} for an illustration of the proof.
	
	\begin{figure}[ht]
		\centering
		\includegraphics{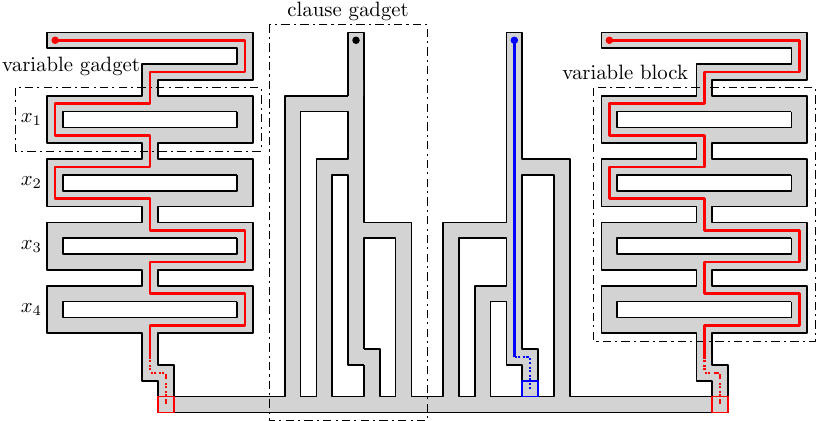}
		\caption{ 	
			Schematic overview of the \NP-hardness reduction. The depicted instance derived from the Boolean \textsc{3Sat} formula $\varphi=(x_1\lor x_2 \lor \overline{x}_3)\land (\overline{x}_2\lor x_3 \lor x_4)$.
			A~sequence that merges the two red particles with $\nicefrac{1}{2}(\diam + b)$ commands corresponds to a variable assignment for $\varphi$.}
		\label{fig:hardness}
	\end{figure}

	For every instance~$\varphi$ of \textsc{3Sat}, we construct a thin polyomino $P_\varphi$ as follows: 
	For~every variable, we insert a variable gadget.
	We join all variable gadgets vertically in row to a \emph{variable block}; we call the top row of each variable gadget its \emph{variable row}.
	For every clause, we construct a clause gadget that contains a left (right) arm for each incident positive (negative) literal in the corresponding variable row and an exit arm in the bottom.
	To obtain~$P_\varphi$, we join all clause gadgets from left to right by a \emph{bottom row} and insert a variable block at the left and right end of this row.
	
	Let $I$ be the instance of~\decProb consisting of~$P_\varphi$ where the top row is filled with particles. We call the two leftmost particles above the variable blocks, the \emph{red} particles and denote the length of the bottom row by $b$. Note that the distance between the red particles is the diameter $\diam$.
	
	\begin{claim*}
		$I$ has a gathering sequence of length $\ell:=\nicefrac{1}{2}(\diam + b)$ if and only if $\varphi$ is satisfiable.
	\end{claim*}
	
	If $\varphi$ is satisfiable, consider a satisfying assignment of $\varphi$ and apply the command sequence that moves the left red particle to the left pixel of the bottom row such that it moves left in the variable row of $x_i$ if $x_i$ is true and right, otherwise. Note that each particle of a clause gadget uses an arm of a variable satisfying the clause and thus ends in the bottom row of~$P_\varphi$.
	Then the command sequence $\langle\Left\rangle^b$ merges all particles. This gathering sequence has length $\ell$.
	
	Now we consider the case that $\varphi$ is not satisfiable.
	We show that in a gathering sequence of length $\ell$, the red particles must merge in the bottom row: They do not merge in one of the variable blocks, otherwise the distance of one particle to the merge location exceeds $\ell$. Moreover, the two particles move symmetrically through the variable blocks for (at least) the first $\ell-b$ commands. Reaching the bottom row of length $b$ on opposite ends, they can only merge within $b$ steps at the left or right end of the bottom row.
	
	Consequently, the gathering sequence is determined by merging the two red
	particles; only the choice of going left or right in each variable gadget is to be determined, yielding a one-to-one correspondence to each variable
	assignment.  Because $\varphi$ is unsatisfiable, in every variable assignment
	there exists a clause that is not satisfied.
	Consider the top particle of this clause, which we call the \emph{blue particle}.
	When traversing the clause gadget, the blue particle does not use any variable arm while the red particle traverses the variable block; because the exit arms of the variable blocks and clause gadgets are on different heights, the blue particle ends one pixel before the bottom row after $\ell - b$ steps. Thus, the red and the blue particles cannot be merged by $b$ commands.
\end{proof}

In fact, a stronger statement holds true: 
As the leftmost pixel of the bottom row is one of only two possible merge locations for a gathering sequence of length $\nicefrac{1}{2}(\diam + b)$, the same reduction shows that the problem remains \NP-complete, even if target locations are~prescribed.

\begin{corollary}
	\decProb is \NP-complete, even with prescribed targets.
\end{corollary}

Furthermore, we obtain another result from the very same reduction. 
Consider an instance of~\decProb where all pixels are filled with particles; this instance has a gathering sequence of length $\nicefrac{1}{2}(\diam + b)$ if and only if $\varphi$ is satisfiable.
This fact implies \NP-completeness of the following problem.

\pagebreak
\noindent\textsc{\underline{Tilt Robot Localization}}

\smallskip
\noindent\textsc{Given:} A polyomino $P$, a sensorless robot $r$ within $P$, and an integer $\ell$.\\
\noindent\textsc{Task:} Decide whether there exists a moving sequence such that we know the robot's location after at most $\ell$ steps.
\medskip

In an instance of this problem, we are given a sensorless robot $r$ within a polyomino, and we wonder whether there exists a command sequence of length~$\ell$ such that we know the position of $r$ afterwards, regardless of the initial position of $r$.
The above observations yield:

\begin{corollary}\label{cor:np_hard_2}
	\textsc{Tilt Robot Localization} is \NP-complete.
\end{corollary}

We note that the general problem of localizing a robot in a known polygonal domain obtained much research in the past~\cite{DudekRW98}.


\section{Algorithmic approaches}\label{sec:gathering:algorithms}
After establishing \NP-completeness of the problem in general, we now focus on strategies that eventually merge particles, i.e., reduces the number of particles in the workspace. 

We start with \emph{simple} polyominoes, i.e., a polyominoes $P$ for which decomposing~$P$ by horizontal lines through pixel edges results in a set of rectangles $\mathcal R$ such that the edge-contact graph $\mathcal{C}(\mathcal{R})$ of $\mathcal R$ is a tree. 
The edge-contact graph of a set of rectangles in the plane contains a vertex for each rectangle and an edge for each side contact; a corner contact does not result in an edge.
A \emph{hole} of a polyomino~$P$ here refers to a maximal set of blocked cells (cells not contained in $P$) that are connected such that there exists a closed walk within $P$ surrounding it.
As usual, simplicity of a polyomino captures the feature of not containing holes.
A \emph{shortest path from a pixel $p$ in $P$ to a rectangle~$R$ in $\mathcal R$} is a shortest path from $p$ to a pixel $q$ in $R$ such that $\dist(p,q)$ is minimal.

\begin{theorem}\label{thm:simpleTwo}
	For any two particles in a simple polyomino~$P$, there exists a gathering sequence of length~$\diam$.
\end{theorem}
\begin{proof}
	Let $\mathcal{R}$ be a decomposition of $P$ into rectangles by cutting $P$ with horizontal lines through pixel edges. Then, because $P$ is simple, the edge-contact graph $\mathcal{C}(\mathcal{R})$ of the rectangles $\mathcal{R}$ is a tree. For an example, consider \cref{fig:Simple}.
	\begin{figure}[htb]
		\centering
		\includegraphics{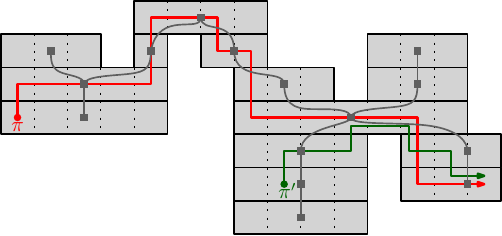}
		\caption{A simple polyomino $P$, and its edge-contact graph $\mathcal{C}(\mathcal{R})$ (in gray). When the red particle $\pi$ moves towards the green particle $\pi'$, $\pi$ and $\pi'$ follow the respective red and green paths. The dotted lines separate the pixels.}
		\label{fig:Simple}
	\end{figure}

	For every $t$, let $R_t$ and $R'_t$ be the rectangles of $P$ containing the two particles $\pi$ and $\pi'$ after applying $t$ commands, respectively. Moreover, let $S_t$ be a shortest path from $R_t$ to $R'_t$ in~$\mathcal{C}(\mathcal{R})$; and let $S_t(1)$ be the successor of~$R_t$ on $S_t$ (if it exists, i.e., $R_t\neq R'_t$).
	\pagebreak
	
	We use the following strategy:
	\begin{description}
	\item[Phase~1:] While $R_t\neq R'_t$, compute a shortest path $S_t$ from $R_{t}$ to $R'_{t}$ in $\mathcal C(\mathcal R)$. 
	Move $\pi$ to $S_t(1)$ via a shortest path in~$P$.
	Update $R_t$ and $ R'_t$. 
	
	\item[Phase~2:] If $R_t=R'_{t}$, merge $\pi$ and $\pi'$ by moving $\pi$ towards~$\pi'$ by a shortest (horizontal) path; note that this sequence merges the particles within $R_t$.
	\end{description}
	We now show that this strategy yields a gathering sequence of length~$D$.
In fact, the resulting sequence has the following property.

	\begin{claim*}\label{clm:1}
		For every $s>t$, the rectangles $R_s$ and $R'_s$ are either equal to $R_t$ or lie in the connected component~$C$ of $\mathcal C(\mathcal{R}\setminus R_t)$ containing $R'_t$.
	\end{claim*}
To arrive at a contradiction, assume that $\pi$ enters a rectangle~$R_s$ ($\neq R_t$) that is not in~$C$.
	Because $\pi$ moves towards~$\pi'$ in every step, there also exists $p>t$ such that $R'_p$ ($\neq R_t$) is not contained in $C$.
	Because $P$ is simple, there exists an $j$ with $p> j\geq t$ such that $R_j=R'_{j}=R_t$. However, because $R_{j}=R'_{j}$, $\pi$ and $\pi'$ merge in $R_j$, i.e., $R_s=R_t$ and $R_p=R_t$; a contradiction.
	
	The claim implies that $\pi$ never re-enters a rectangle, i.e., it moves to an unseen rectangle in every step of Phase~1. Because~$P$ is finite, $\pi$ and $\pi'$ eventually meet in some rectangle which ends Phase~1 and therefore merge in Phase~2. 
	
	Moreover, for every $t$, the merge location and $R_t'$ lie in $C$ or are equal to $R_t$. Consequently, in every step, $\pi$ moves towards the merge location on a shortest path.
	Because a shortest path is at most of length $\diam$, the length of the gathering sequence is bounded by $\diam$.
This implies that the merge location and $R_t'$ lie in~$C$ or are equal to $R_t$. Consequently,
in every step, $\pi$ moves towards the merge location on a shortest path and thus that the gathering sequence is at most of length $D$.
\end{proof}

In the remainder, we refer to the strategy used to prove \cref{thm:simpleTwo} as \emph{Dynamic-Shortest-Path} (abbrv. DSP):
Move one particle towards the other along a shortest path; update the shortest path if a shorter one exists.
The example depicted in \cref{fig:ex1} shows that DSP may perform significantly worse in non-simple polyominoes.
	
\begin{proposition} 
		In polyominoes with holes, the strategy DSP may not yield a gathering sequence of length $\mathcal{O}(\diam)$. 
	\end{proposition}
\begin{proof}
	Consider a polyomino with \( H \) holes, each of width \( w \) and height \( h \), as depicted in \cref{fig:ex1} for \( H = 4 \). Place two particles, \( \pi \) and \( \pi' \), as illustrated. 
	By the symmetry of $P$, the distance between
	 $\pi$ and $\pi'$
	  decreases for the first time when one of them is at the left or right side of $P$. 
	  Therefore, when the red particle \( \pi \) moves towards the green particle \( \pi' \) by a Dynamic Shortest Path (DSP), they traverse the entire bottom path.
	  Consequently, we obtain a lower bound of \( H(6h + w) + 3 \) for the gathering sequence \( C \).
	Note that the diameter is bounded by $\diam\leq (H-2)w+6h+2w+4=Hw+6h+4$. 
	\begin{figure}[ht]
		\centering
		\includegraphics{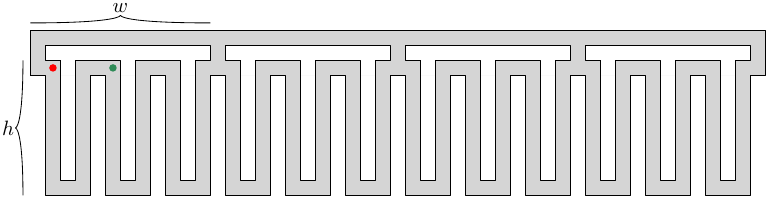}
		\caption{When the red particle $\pi$ moves towards the green particle $\pi'$ by shortest paths, $\pi$ visits the entire bottom path.}
		\label{fig:ex1}
	\end{figure}

	By choosing the height $h:=\nicefrac{cw}{6}$ for some constant $c\geq H$, the ratio of $|C|$ and $|D|$ becomes arbitrarily large: 
	\begin{displaymath}
		\frac{|C|}{|D|}\geq \frac{Hcw+Hw+3}{Hw+cw+4}\geq \frac{H(c+1)}{H+c+1}\geq\frac{H}{2}.\qedhere
	\end{displaymath}
\end{proof}

Nevertheless, DSP always merges two particles, as we show in the following.

\begin{proposition}
	For every polyomino $P$ with $n$ pixels and diameter $D$ and every configuration with two particles, DSP yields a gathering sequence of length $\mathcal{O}(nD)$.
\end{proposition}

\begin{proof}
	Let $\pi$ follow $\pi'$. We show that within $n$ commands, their distance $\Delta$ decreases at least by one.
	Note that if the shortest path must be updated, the distance decreases.
	
	Consider a sequence $C$ of $\ell$ steps, in which $\Delta$ remains constant. In $C$, $\pi'$ has no collision and the shortest path is only updated when $\pi$ reaches the end of the current shortest path, i.e., a previous position of $\pi'$. Let $p_0$ and $p_1$ denote the initial position of $\pi$ and $\pi'$, respectively. When $\pi$ reaches $p_i$, $p_{i+1}$ denotes the position of $\pi' $.
	Let $v$ be the coordinate vector from $p_0$ to $p_1$. Because $\pi$ has no collision, $v$ is the coordinate vector from $p_i$ to $p_{i+1}$ for every $i$. Note that $\pi'$ must have a collision when moved~$n$ times in direction $v$ (because $P$ ends). Moreover, $\pi$ reaches $p_{i+1}$ from $p_i$ after at most~$D$ commands. Consequently, $\ell\leq nD$.
\end{proof}

By using a different strategy, we obtain a better bound:
The strategy \emph{Move-To-Extremum} (abbrv. MTE) iteratively moves an extreme particle (e.g., bottom-leftmost) to an opposite extreme pixel (e.g., top-rightmost) along a shortest path.

\begin{theorem}\label{thm:NonSimpleTwo}
	For any two particles in a polyomino $P$, MTE yields a gathering sequence of length at most $\diam^2$.
\end{theorem}
\begin{proof}
	Let $q$ be the top-rightmost pixel of $P$. To merge the two particles in $q$, our strategy is as follows: 
	Identify the particle $\pi$ that is bottom-leftmost. Apply a command sequence that moves $\pi$ to $q$ on a shortest path. Repeat.
	
	\begin{claim*}
		In each iteration, the sum of the distances $\Delta$ of the two particles to $q$ decreases.
	\end{claim*} 
	Note that $\Delta$ decreases when the other particle $\pi'$ has a collision. If $\pi'$ had no collision, there exist a pixel that is higher or more to the right than $q$, contradicting the choice of $q$.
	Consequently, $\Delta$, which is at most $2D$ at start, decreases at least by~1 for every $D$ steps. Hence, after $\mathcal{O}(D^2)$ steps, $\Delta$ is reduced to~0.
\end{proof}
Note that there exist polyominoes, for example a simple square, where the number of pixels~$n$ is in $\Omega(D^2)$. Therefore, our result significantly improves on the best previously known bound of $\mathcal{O}(n^3)$ that was shown in~\cite{gathering_swarm}.

\smallskip
Finally, we show that a shortest gathering sequence for two particles in a  polyomino with holes may need to exceed $D$.
\begin{proposition}\label{thm:lowerBound}
	Let $P$ be a polyomino with two particles. A shortest gathering sequence may be of length $\nicefrac{3}{2}\,\diam-\mathcal{O}(\sqrt{\diam})$.
\end{proposition}

\begin{proof}
	Let $h\in \mathbb{N}$.
	Consider the polyomino $P$ illustrated in \cref{fig:lowerbound-non-simple} which consists of the bottom row and $S$ many chimneys of height~$h$ and length $2h+4$. Because $P$ consists of $(2h+4)S+6S$ pixels, it has a diameter of $\diam = (h+5)S$. We set $S=2h+4$.
	The two particles $\pi_1$ and $\pi_2$ have an initial distance of $\diam$ such that the number of chimneys to the left of $\pi_1$ and to the right of $\pi_2$ is $\nicefrac{1}{2}\,(h-1)$.

	\begin{figure}[htb]
		\centering
		\includegraphics{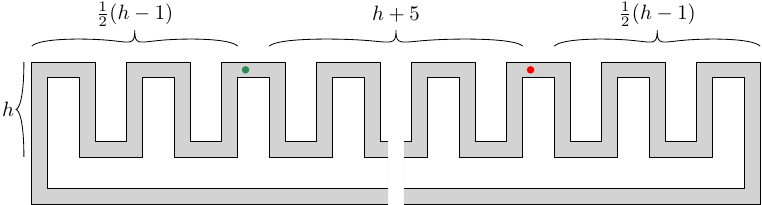}
		\caption{A polyomino consisting of a base and $S$ many chimneys.}
		\label{fig:lowerbound-non-simple}
	\end{figure}
	Due to the symmetry of $P$ and the placement of the particles, the distance $\Delta$ of~$\pi_1$ and~$\pi_2$ cannot decrease within the first 
	$\nicefrac{1}{2}\,(h-1)(2h+4)-1= \nicefrac{\diam}{2} -\mathcal{O}(\sqrt{D})$ commands.
	Without loss of generality, we assume that $\pi_1$ is at the left side of $P$ right before $\Delta$ decreases for the first time.
	We call this configuration \emph{west-touch}.
	Starting from the west-touch configuration, the best merge location is in the top of the leftmost chimney. 
	
	Hence, the distance of $\pi_2$ to the merge location is $\diam - h =\diam -\mathcal{O}(\sqrt{D})$. 
	Consequently, a~gathering sequence is at least of length $\nicefrac{\diam}{2} +D-\mathcal{O}(\sqrt{D})$. 
\end{proof}

\subsection{Reducing the number of particles significantly}\label{sec:tiltgather:cornering}

Now we show how to significantly decrease the number of particles by few commands to a parameter proportional to the complexity of the polyomino, namely the number of convex corners. 
This is particularly relevant for establishing \emph{oblivious} gathering strategies that are capable of merging all particles efficiently, even if their initial configuration is not known. 
This is of interest in practical applications as it may be costly to determining or knowing the location of each individual particle beforehand. 
These settings are equivalent to the situation where each cell contains a particle, and thus, \textsc{Tilt Robot Localization}, see~\cref{sec:gathering-hard}.

\begin{lemma}\label{lem:ReduceParticles}
	Let $P$ be a polyomino with diameter $D$ and $k$ convex corners.
	For every configuration $A\in\Part$, there exists a command sequence of length $2\diam$ which transforms $A$ to a configuration $A'\in\Part$ such that $|A'|\leq \nicefrac{k}{4}$.
\end{lemma}
\begin{proof}
	We distinguish four types of convex corners; northwest (NW), northeast
(NE), southwest (SW), southeast (SE). By the pigeonhole principle, one of the
types occurs at most $\nicefrac{k}{4}$ times; without loss of generality, let
this be the NW corners.
	
	We show that after applying the sequence $\langle\Left,\Up\rangle^\diam$, every particle lies in a NW corner:
	Consider a particle~$\pi$ in pixel~$p$. Unless $\pi$ lies in a NW corner, it moves for at least one command  in $\{\Left,\Up\}$. Because $P$ is finite, there exists an $\ell$ large enough such that $\pi$ ends in a NW corner $q$ when the command sequence $\langle\Left,\Up\rangle^\ell$ is applied, i.e., there exists an $pq$-path consisting of at most $\ell$ commands of types  \Left and  \Up, respectively. Because a monotone path is a shortest path, it holds that  $\ell\leq \diam$.
\end{proof}

\subsection{Upper bounds on the length of gathering sequences}

Finally, we provide straightforward general upper bounds for the length of gathering sequences in simple as well as non-simple polyominoes; they follow easily by combining different preceding results.

\pagebreak
For simple polyominoes, combining \cref{thm:simpleTwo,lem:ReduceParticles} yields:
\begin{corollary}
	For a set of particles in a simple polyomino $P$ of diameter $D$ and $k$ convex corners, there exists a gathering sequence of length $\mathcal{O}(k\diam)$.
\end{corollary}

However, the combination of \cref{thm:NonSimpleTwo,lem:ReduceParticles} imply the following for polyominoes with holes:

\begin{corollary}\label{thm:NonSimpleMany}
	For any set of particles in a non-simple polyomino $P$ of diameter~$D$ and $k$ convex corners, there exists a gathering sequence of length at most $\mathcal{O}(k \diam^2)$.
\end{corollary}

By analyzing cuboids instead of rectangles, six directions of motion instead of four, and corners in eight quadrant directions instead of four, it seems plausible to obtain analogous results for three-dimensional settings.

	\section{Reinforcement Learning}\label{sec:tiltgather:deep-learning-impl}

When looking at the pixilated graphics in \cref{fig:tiltgather:over_time_graphic}, we are reminded of classical arcade games.
In 2013, Mnih et al.~\cite{mnih2013playing} proposed an algorithm based on deep reinforcement learning that was able to learn how to play Atari games and eventually reach the level of professional players~\cite{mnih2015human}.
Because such algorithms have now become freely available and reasonably easy to use, e.g., \emph{Stable Baselines}~\cite{stable-baselines}, we can let them `play' \textsc{Min-Gathering}.

In reinforcement learning, an \emph{agent} is exposed to an \emph{environment} in which it has to achieve an objective by performing a sequence of \emph{actions} that manipulate the \emph{state} of the environment.
Upon each action, the environment can yield a \emph{reward} (or penalty), or it can terminate.
During repeated trials, the agent has to learn a \emph{policy} that maximizes the (expected) total reward by choosing actions based on \emph{observations}, i.e., the perceived states.

Let us denote the state space by $\mathcal{S}$, the actions by $\mathcal{A}$, the transition function by ${T: \mathcal{S}\times \mathcal{A} \rightarrow \mathcal{S}}$, and the reward function by $R: \mathcal{S}\times \mathcal{A}\rightarrow \mathbb{R}$.
For simplicity, this notation neglects probabilistic elements that are common in other use cases.
Starting at an initial state $s_0 \in \mathcal{S}$, the agent selects an action $a_0\in \mathcal{A}$ based on the probabilities of the policy $\pi: \mathcal{S}\times \mathcal{A} \rightarrow [0, 1]$, resulting in the reward $r_0=R(s_0, a_0)$, and leading to the next state $s_1=T(s_0, a_0)$.
This is repeated until a final state is reached, after which the total reward $\sum_i r_i$ is evaluated.
We call such a sequence a \emph{period}.

We use a \emph{Convolutional Neural Network} (CNN) on an image of the particle locations as a policy.
The construction and training can be performed by readily available algorithms such that we only need to provide the environment, but it is useful to understand the basics of the algorithms before we design the environment.
While we could simply display the particles as a matrix and give a reward for successful gathering, the algorithms could not learn efficiently from such an environment.
In the following, we give a short introduction into the inner workings of these algorithms.

During training, we need to optimize the (initially random) network parameters to (incrementally) yield better actions.
A fundamental problem in reinforcement learning is that we need a sequence of actions to gain a reward, but the policy needs to be trained for individual actions.
Often, we even need to perform some penalized actions along the way.
How can we compute the advantage of actions not only based on their direct reward or penalty but based on their long-term influence, such that the policy can improve its output?
This is known as the \emph{credit assignment problem}, and a common strategy is to consider all future rewards (reduced by some \emph{discount factor}) for each action.

\pagebreak
Let us take a look at the simple REINFORCE algorithm~\cite{williams1992simple} as an example:
\begin{enumerate}
  \item Run multiple periods with the neural network (that initially only returns random values) and compute for each returned action the gradients of the network parameters which increase the probability of this action.
  \item Compute the advantages of the performed actions.
    For each period $(s_0, a_0, r_0), (s_1, a_1, r_1),$ $\ldots, (s_n, a_n, r_n)$, the advantage of performing $a_i$ in state $s_i$ is determined by $A(s_i, a_i)= \sum_{j=i, \ldots, n} \gamma^{j-i} \cdot R(s_j, a_j)$, where $\gamma\in (0, 1)$ is the discount factor.
    We standardize the advantages by subtracting the mean over all periods and dividing by the standard deviation.
    Positive standardized advantages now correspond to actions that lead to above average rewards.
  \item Multiply the gradients of the first step with the standardized advantages of the second step, and use their mean to perform a \emph{Gradient Ascent} step on the neural network.
    The more advantageous an action has been, the stronger it gets reinforced.
  \item Repeat this procedure until the policy performs sufficiently well.
\end{enumerate}

In our implementation, we use the more advanced \emph{Proximal Policy Optimization} (PPO) algorithm~\cite{schulman2017proximal}, which actually learns not only the best action but also the advantage.
The~differences between PPO and REINFORCE are significant in practice, but for designing a reasonably good reward function as in \cref{sec:tiltgather:drl:reward}, understanding the idea of the REINFORCE algorithms should suffice.
For a deeper understanding, we refer the curious reader to the extensive current literature, such as~\cite{graesser2019foundations}.

We design the environment for \textsc{Min-Gathering} as follows:
\begin{itemize}
    \item The observations are images of the particle locations scaled to $84\times 84$ pixels with a maximum filter for all instances.
      This is a common resolution keeping the CNN reasonably small, which not only speeds up the computations but can also help the CNN to generalize.
    \item We fill the environment completely with particles.
      A motion sequence that gathers all particles can be applied to any configuration.
      However, it is also possible to use a concrete configuration as the initial state.
    \item Each action is repeated automatically a fixed number of times, also called \emph{frame skipping}.
      This speeds up the learning process drastically, as a few random (repeated) actions can result in visible gathering progress.
      Otherwise, random actions have only a low chance of making notable progress.
    \item We not only provide the four basic motions as actions but also add diagonal motions, which are simulated by two basic motions in random order.
      This drastically improved the performance in preliminary experiments, especially with much \emph{frame skipping}.
      Diagonal movements are also a common pattern in the solutions.
      Providing them directly speeds up the learning process.
    \item The lowered resolution and the frame skipping makes it difficult to gather particles to a single location.
      Thus, we consider the particles as gathered if they are within a radius of \num{10} steps of an extreme point.
      The final gathering is then performed by the heuristic \textsc{MinSumToExtremum}, which only needs around \num{15} additional motions in our experiments as the particles are already close to an extreme point.
    \item If the particles are not gathered after \num{500} motions for \texttt{Corridor}, \num{800} motions for \texttt{Capillary}, or \num{3500} motions for \texttt{Brain}, the period gets terminated.
      This prevents the algorithm from spending too much time if the agent `gets lost'.
      These limits have been chosen based on the performance of the classical algorithms.
    \item The design of the reward function is the most important part; discussed in \cref{sec:tiltgather:drl:reward}.
    \item The extraction of the best solution directly from the learning process; see~\cref{sec:tiltgather:drl:implementation}.
\end{itemize}

\subsection{Reward}\label{sec:tiltgather:drl:reward}
Giving a reward only after all particles are gathered is not practical because every period that does not gather all particles looks equally bad.
Using this method, a command sequence that is capable of getting the particles at least close by looks as bad as a command sequence that does nothing at all.
This implies that as long as the particles are not gathered by chance, which is very unlikely, all actions are classified as bad.
If all action are equally bad, the policy cannot do anything but perform random motions.

A straightforward alternative is to give a reward every time the diameter of the particle swarm is reduced.
Of course, we only give a reward if the all-time minimal diameter is reduced; otherwise, repeatedly growing and shrinking the particle swarm may be learned as a `good' strategy.
To encourage a short gathering sequence, we can additionally give a small penalty for every action.

An issue with this reward function is that it is time-consuming to compute after every step.
We know that extreme points in the maze are good gathering locations.
Instead of computing the diameter, we could compute the maximal distance of a particle to an extreme point.
Let $E$ be the set of extreme points, then the maximal distance is defined as $\min_{e \in E} \max_{p\in P} \text{dist}\left(p, e\right)$.
This gives us the option to choose any extreme point, but it must be the same for all particles.
The distance of each location to an extreme point is fixed and can be computed ahead of time, thereby making this computation just a lookup for every~particle.

For instances like \texttt{Brain}, we first need to escape some recesses before we can minimize the distance.
This has the same issue as before of being hard to achieve through random moves.
Here we can support the detection of progress not only by providing reward for minimizing the maximal distance but also by providing reward for minimizing the mean~distance.

As the distances can vary strongly for different instances, we normalize the rewards such that each of the two reward components can only give an accumulated maximal reward of $1$, and the motion penalty an accumulated penalty of at most $-1$.
We normalize the reward of minimizing the maximal or mean distance by $1$ by dividing through the initial distance.
The rewards are normalized via dividing by the initial maximal resp.\ mean distance.
The motion penalty is normalized by the motion limit $L$ such that each motion gives a penalty of~$-\nicefrac{1}{L}$.

\begin{figure}[htb]
	\centering
  \begin{subfigure}[b]{0.45\columnwidth}
    \includegraphics[width=0.8\columnwidth]{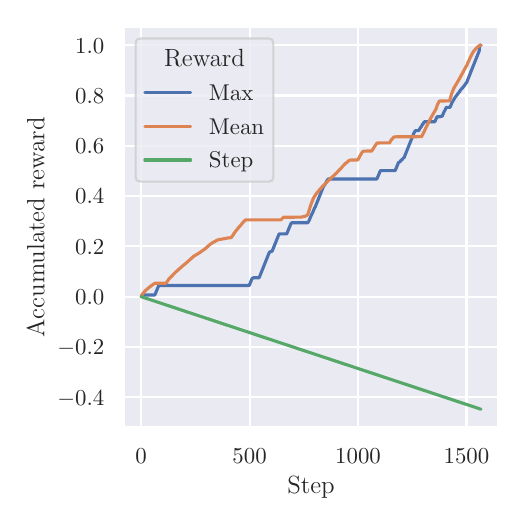}
    \caption{Accumulated rewards.}
  \end{subfigure}
  \begin{subfigure}[b]{0.45\columnwidth}
    \includegraphics[width=0.8\columnwidth]{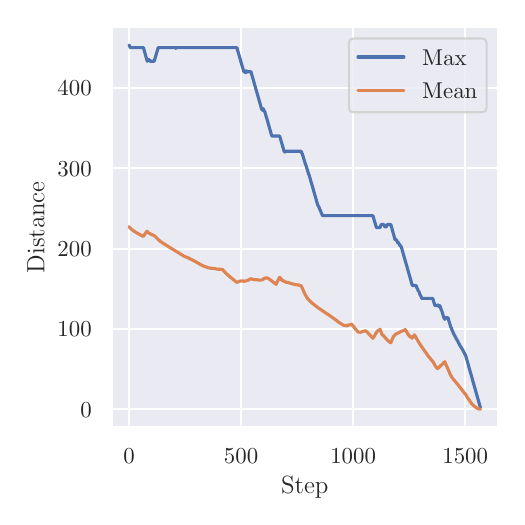}
    \caption{Maximal and mean distance.}
  \end{subfigure}
  \caption{Rewards and distances during the best movement sequence of the \texttt{Brain} instance, computed with reinforcement learning.
  We can see that the maximal distance, which we want to minimize, stagnates during the first \num{500} steps.
  Maximal and mean distances must both overcome local minima.}\label{fig:tilt:gather:drl:rewardexample}
\end{figure}
An example for the \texttt{Brain} instance is given in \cref{fig:tilt:gather:drl:rewardexample}.
We see that the mean distance supports the maximal distance well as it is more continuous than the maximal distance.
The maximal distance only markedly decreases after \num{500} steps, which is too long without proper feedback.
The maximal and mean distances must overcome local minima until the particles are gathered, showing that a na\"ive local search would not be successful.
Even if a local search were to succeed, it is unlikely that the gradients would always point in a direction that leads to a short gathering sequence.
The strength of the reinforcement learning approach, however, is that it automatically improves the reward function.

\subsection{Implementation}\label{sec:tiltgather:drl:implementation}
In this section, we show how to easily implement the optimizer.
Modern reinforcement learning libraries require barely any knowledge of reinforcement learning or neural networks.
Of course, some knowledge and experience is useful, especially for designing a good reward function.
This is comparable with modern MIP-solvers that can be used without much understanding of the underlying techniques, but creating a good formulation that can be quickly solved is more difficult.
A simple implementation that finds the shortest action sequence to achieve the objective can look as follows:
\begin{lstlisting}[language=Python]
# Create a learning environment
class LearningEnv(gym.Env):
  def __init__(self, simulation, limit, repeat):
    self.simulation = simulation
    self.rewards = Rewards(simulation, limit)
    self.repeat = repeat
    self.min_solution = None
    # Define input and output for the neural network
    self.action_space = gym.spaces.Discrete(simulation.no_actions)
    self.observation_space = gym.space.Box(low=0, high=255, 
                   shape = simulation.image_shape, dtype=np.uint8)

  def step(self, action):
    # Perform a step in the simulation and give feedback.
    for i in range(self.repeat):
      self.simulation.step(action)
    observation = self.simulation.as_image()
    reward, abort = self.rewards.eval_state()
    if self.simulation.is_gathered():
      abort = True
    new_sol = self.simulation.history
    if self.min_solution is None or
        len(self.min_solution) > len(new_sol):
      self.min_solution = new_sol

    return observation, reward, abort, {}

  def reset(self):
    # Reset simulation and rewards for the next try.
    self.simulation.reset()
    self.rewards.reset()

# Load instance
simulation = Simulation('my_instance.json')
# Allow up to 1000 steps and automatically repeat four times
env = LearningEnv(simulation, limit=1000, repeat=4)
# Automatically resize observations to 84x84
resized_env = ResizeObservation(env, shape=(84, 84))
# Automatically build neural network (CNN)
model = PPO('CnnPolicy', resized_env)
# Try for 300 000 steps
model.learn(total_timesteps=300_000)

# Output best solution
print("Solution:", env.min_solution)
\end{lstlisting}
This code uses Stable Baselines \num{3}~\cite{stable-baselines}, but interfaces in state-of-the-art machine learning libraries are highly volatile.
In this case, we are using the PPO algorithm, but it can be replaced by multiple other algorithms.

This implementation only misses two problem-specific details: the simulation and the reward function.
The simulation needs to be able to perform a sequence of individual actions encoded by discrete numbers (e.g., $0$ for up, $1$ for right, \ldots), return the current states as an image or a matrix, detect if the objective has been achieved, remember the corresponding solution, and reset to the initial state.
The reward function only needs to compare the current and the last state and rate the change or abort the current solution process, e.g., by a length limit, if it is not promising.

Note that this approach actually uses the training phase of the neural network to find a solution.
Usually, one trains the neural network to use it afterward.
As our environment does not change and the sequences remain valid, simply using the best encountered solution is the better strategy.

\section{Evaluation}
In this section, we evaluate the performance of the following approaches on practical instances.

\begin{itemize}
  \item The approach \textsc{StaticShortestPath} (SSP) iteratively merges pairs of particles by moving one to the position of the other, along a shortest path, see Algorithm~2 in~\cite{gathering_swarm}.

  \item  The approach \textsc{DynamicShortestPath} (DSP), as described in \cref{sec:gathering:algorithms}.

  \item The approach \textsc{MoveToExtremum} (MTE), as described in \cref{sec:gathering:algorithms}.
    Among the four possible commands, we choose an extremum that minimizes the initial sum of distances to both particles.
  
  \item The heuristic \textsc{MinSumToExtremum} (MSTE) generalizes the idea of \textsc{MTE}.
    It selects an extremum with the smallest initial sum of distances to all particles and iteratively performs a command that decreases this sum the most.
    If no command decreases the sum, two particles are selected and merged by MTE\@. Afterward, MSTE resumes.

  \item Additionally, we evaluate the machine learning approach \textsc{ReinforcementLearning}~(RL), as described in \cref{sec:tiltgather:deep-learning-impl}.
    We use the default parameters for all mazes and only vary the motion repetitions, limits, and time steps.
    The \texttt{Corridor} and the \texttt{Capillary} maze are trained over \num{300000} time steps with a frame skipping of \num{4}.
    The \texttt{Brain} maze is trained over \num{600000} time steps with a frame skipping of \num{16}.
\end{itemize}

The experiments were performed on over \num{130} random particle configurations with \num{1000} particles in each environment.
Every configuration was solved by all strategies to ensure comparability.
The preprocessing used a random direction to move the particles into corners.
We used a workstation equipped with an AMD Ryzen 7 1700 CPU with $8\times \SI{3.0}{\GHz}$ and \SI{32}{\giga\byte} memory, and an Nvidia GTX 1050 Ti GPU with \SI{4}{\giga\byte} memory.

We compare the combinatorial strategies \textsc{SSP}, \textsc{DSP}, \textsc{MTE}, \textsc{MSTE} and their options.
For these strategies, we evaluate the options of (1) choosing a pair uniformly at random, or (2) choosing the pair with maximal distance.
Additionally, we analyze the advantage of using the preprocessing strategy that moves all particles to corners, as described in \cref{sec:tiltgather:cornering}.

\paragraph*{Results}
\cref{fig:tiltgather:combinatorialAlgs} shows that \textsc{MSTE} performs best on average, and that the pair selection and preprocessing options only have a small influence on it.
Only for the small \texttt{Corridor} instance, \textsc{SSP} with most distanced pairs and preprocessing performs visibly better.
For the shortest path strategies \textsc{SSP} and \textsc{DSP}, the most distanced pairs show a significant advantage for the first two environments, while it has only a small influence on \textsc{MTE} and \textsc{MSTE}.
\begin{figure}[tb!]
  \centering
  \begin{subfigure}[t]{0.3\columnwidth}
    \includegraphics[width=\columnwidth]{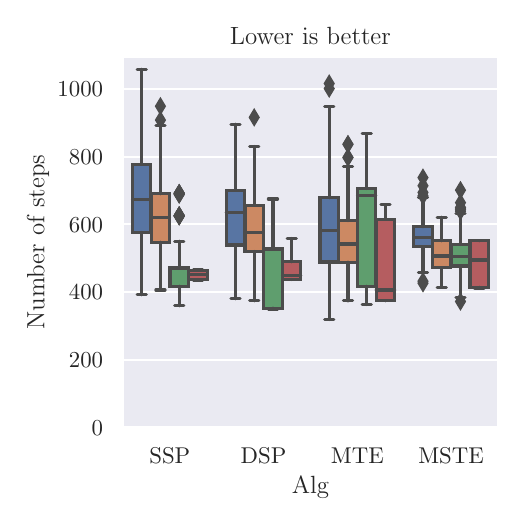}
    \caption{Corridor}
  \end{subfigure}
  \begin{subfigure}[t]{0.3\columnwidth}
    \includegraphics[width=\columnwidth]{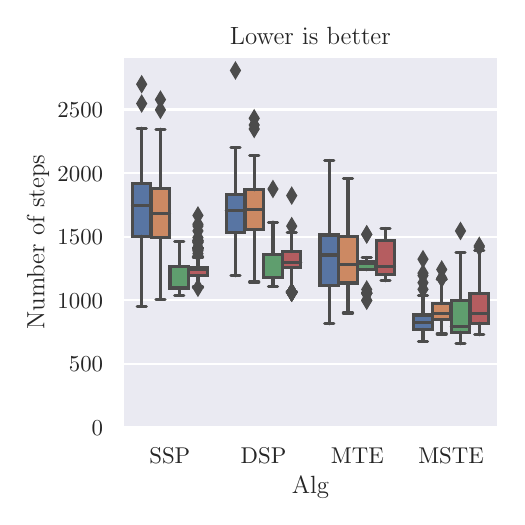}
    \caption{Capillary}
  \end{subfigure}
  \begin{subfigure}[t]{0.3\columnwidth}
    \includegraphics[width=\columnwidth]{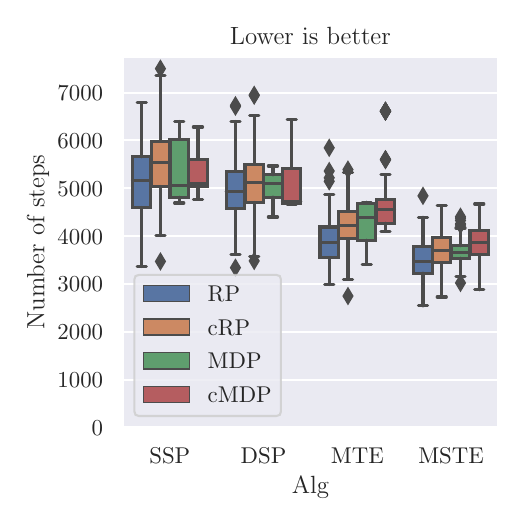}
    \caption{Brain}
  \end{subfigure}
  \caption{Comparison of the combinatorial algorithms with different pair selections (random (RP) or most distanced (MP)) and optional corner preprocessing (cRP resp.\ cMP).}\label{fig:tiltgather:combinatorialAlgs}
\end{figure}

We also experimented with optimal pair merging sequences computed by an $A^*$-algorithm, but in preliminary experiments we encountered worse results at a higher computational complexity such that we ignored this approach for the final experiments.
An explanation for the worse results could be that, e.g., \textsc{MTE} and \textsc{MSTE} are guiding many more particles than just the pair to the extreme position and, thus, are more efficient for gathering all particles even if the strategy may be suboptimal for just the pair.

\begin{figure}[htb]
	\centering
  \begin{subfigure}[t]{0.45\columnwidth}
  \includegraphics[width=\columnwidth]{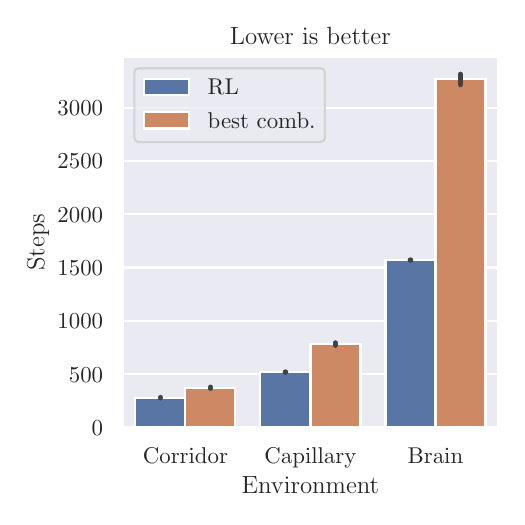}
  \caption{RL vs.\ best combinatorial algorithm.}\label{fig:tiltgather:rlvsbest}
  \end{subfigure}
  \begin{subfigure}[t]{0.45\columnwidth}
    \includegraphics[width=\columnwidth]{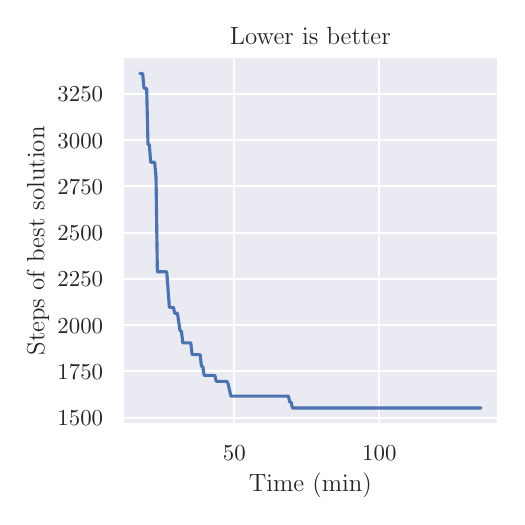}
    \caption{Improvement of RL for \texttt{Brain} over time.}\label{fig:tiltgather:rltime}
  \end{subfigure}
  \caption{The reinforcement learning approach compared to the best results of the combinatorial algorithms. The reinforcement learning approach shows to be superior especially for the complex environments. It already yields superior solutions after a few minutes.}
\end{figure}
When comparing the best solution of the previous algorithms with the solution returned by the reinforcement learning approach in~\cref{fig:tiltgather:rlvsbest}, the reinforcement learning approach shows clear superiority despite assuming a fully filled environment.
Especially for the \texttt{Brain} environment, the reinforcement learning approach yields command sequences of less than half the length of the other algorithms.
The reinforcement learning approach needs over two hours to train, compared with just a few minutes for the execution of \textsc{MSTE}, but already after a few minutes it yields superior solution that improves further over time as can be seen in~\cref{fig:tiltgather:rltime}.
Additionally, the runtime can be improved further by using parallel agents, as PPO supports parallel optimization.

Overall, the reinforcement learning approach is superior and can deal much better with the complex particle configurations than our combinatorial algorithms.
Contrary to the combinatorial algorithms, it can also easily utilize parallelization.
In his master's thesis, Konitzny~\cite{makonitzny} performs a deeper and more extensive analysis of this technique and achieves even better results by lower level optimizations, e.g., replacing the activation neurons.
Contrary to our work, he uses trained agents to perform the gathering; this also allows non-deterministic environments with continuous physics, showing the flexibility of this~approach.

\paragraph*{Oblivious Merging}\label{sec:tiltgather:experiments:oblivious_merging}
In practice, it may be costly to determine the position of the individual particles;
therefore, \emph{oblivious} approaches that do not need this information can be of interest.
Such a setting is equivalent to the situation where initially, each pixel contains a particle;
a gathering sequence for all particles is certainly a gathering sequence for any other (partial) initial distribution of particles.
Recall~\cref{cor:np_hard_2}, implying that this problem remains \NP-complete.
In order to estimate the cost of this restriction in practice, we study how the number of populated grid cells behaves over time, depending on the initial number of particles; see \cref{fig:tiltgather:over_time_plot,fig:tiltgather:over_time_graphic,fig:tiltgather:over_time_graphic2}.

\begin{figure}
  \centering
  \includegraphics[width=0.5\columnwidth]{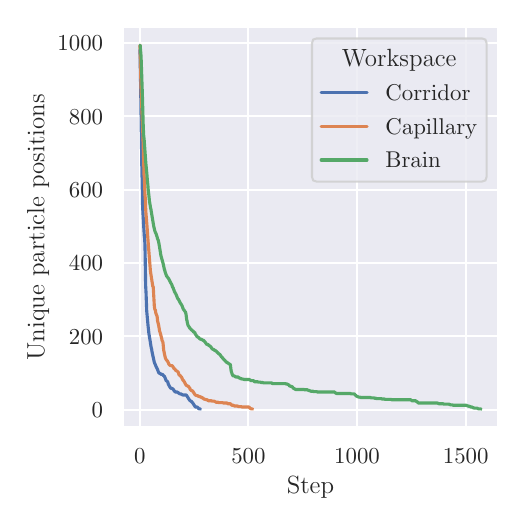}
  \caption{Number of particle groups over time using \textsc{RL} for \num{100} random configurations each of \num{1000} particles.
  For all workspaces, the number of unique particle positions drops very quickly, and collecting the last \num{100} particles comprises the majority of the steps.}\label{fig:tiltgather:over_time_plot}
\end{figure}
As the number of populated grid cells decreases very sharply in the beginning and almost all steps are used to merge the few remaining groups of particles, we conclude that missing knowledge of the position of individual particles has negligible cost for uniform distributions.

\begin{figure*}[tbhp!]
	\newcommand{\fwidth}{0.3\columnwidth}
	\centering
	\includegraphics[width=\fwidth]{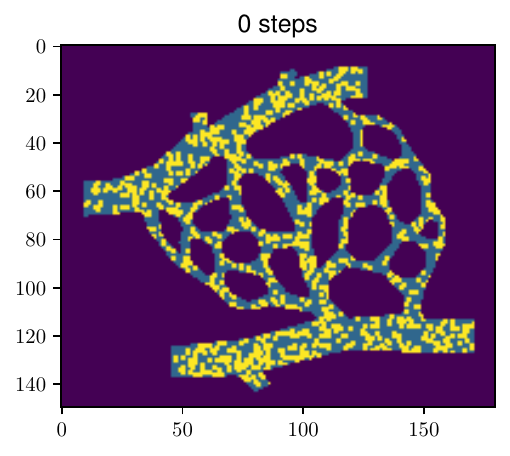}
	\includegraphics[width=\fwidth]{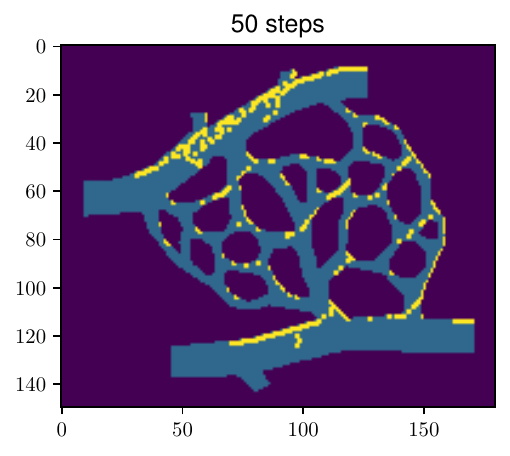}
	\includegraphics[width=\fwidth]{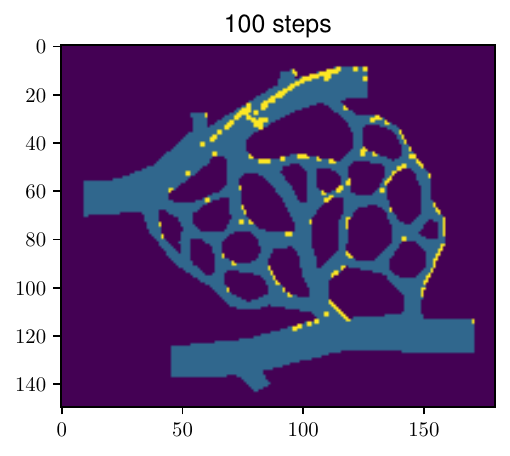}
	\includegraphics[width=\fwidth]{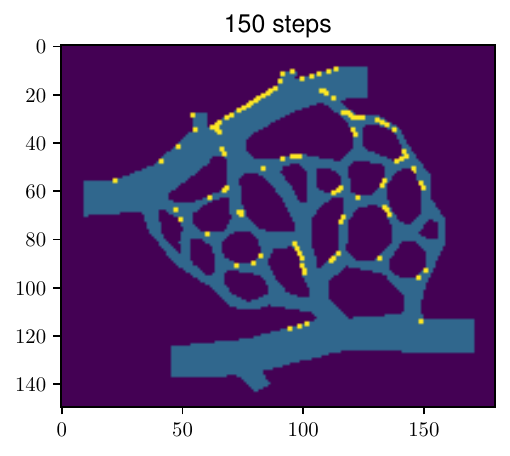}
	\includegraphics[width=\fwidth]{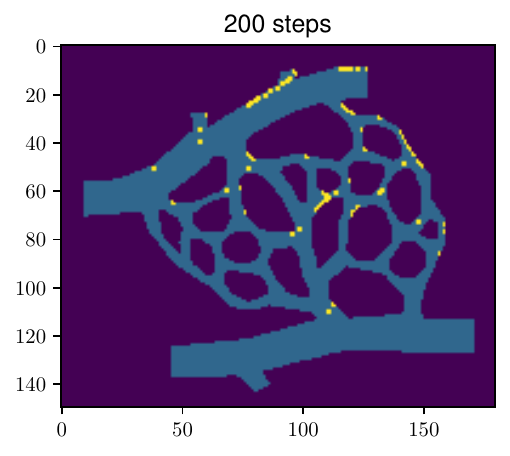}
	\includegraphics[width=\fwidth]{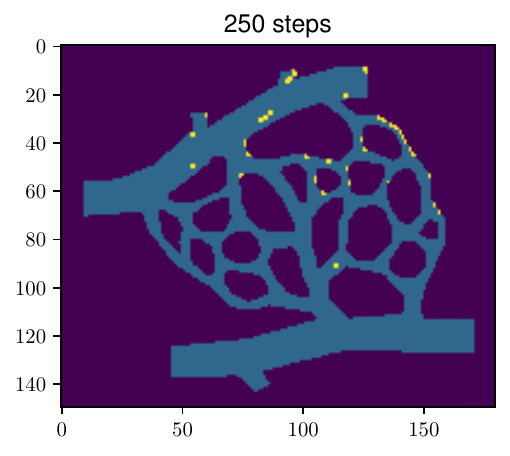}
	\includegraphics[width=\fwidth]{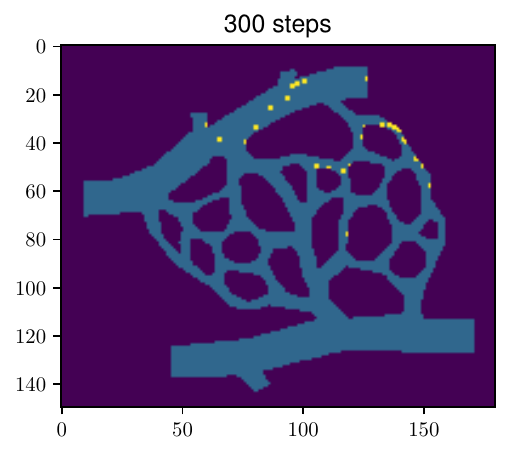}
	\includegraphics[width=\fwidth]{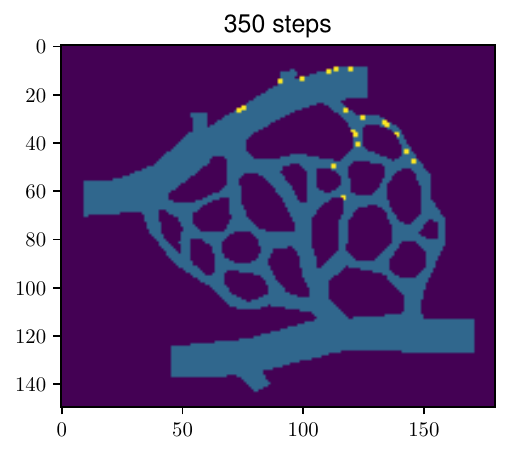}
	\includegraphics[width=\fwidth]{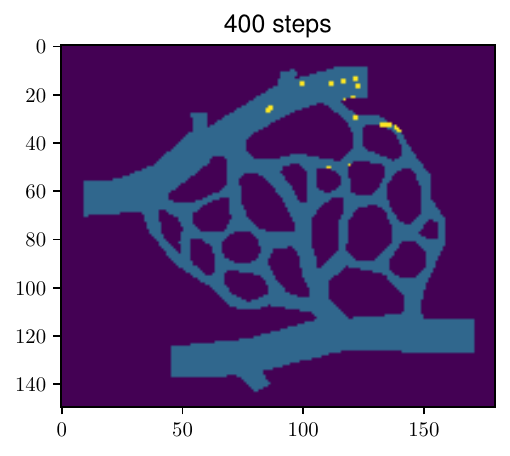}
	\includegraphics[width=\fwidth]{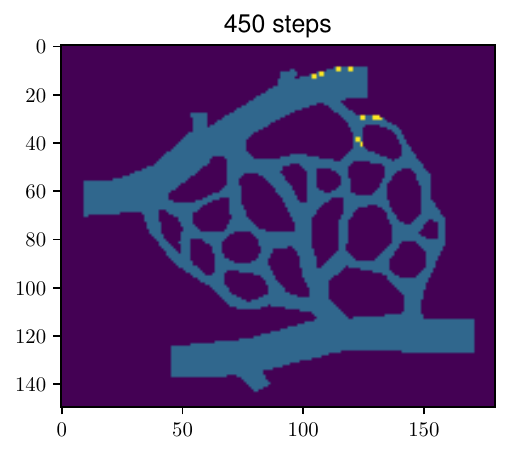}
	\includegraphics[width=\fwidth]{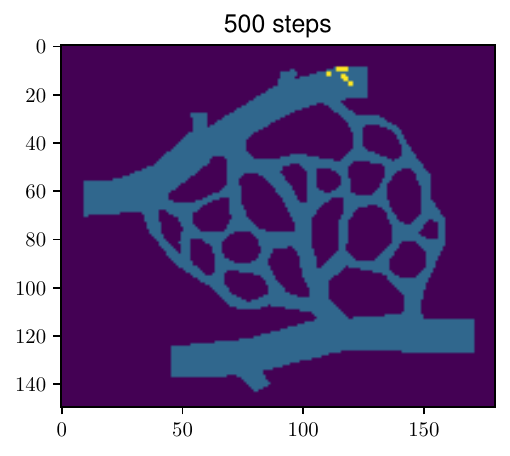}
	\includegraphics[width=\fwidth]{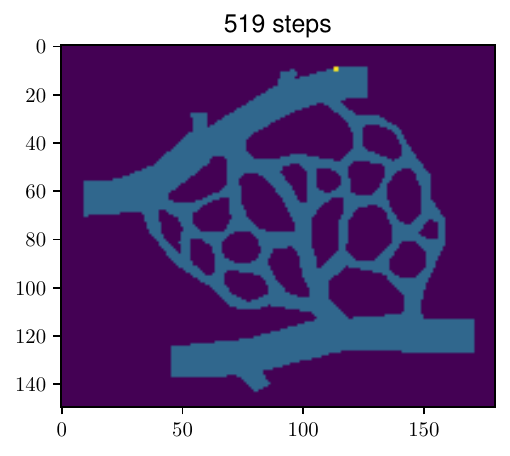}
	\caption{The process of gathering \num{1000} particles in \texttt{Vessel} with \textsc{RL}. The visualization uses a maximum filter to improve visibility. Gathered particles aggregate to a single particle. We see that the particles quickly collapse to a few clusters.}\label{fig:tiltgather:over_time_graphic}
  \end{figure*}
  
  \begin{figure*}[tbhp]
	\newcommand{\fwidth}{0.3\columnwidth}
	\centering
	\includegraphics[width=\fwidth]{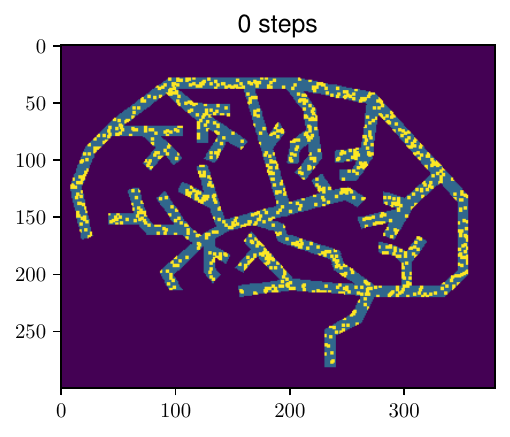}
	\includegraphics[width=\fwidth]{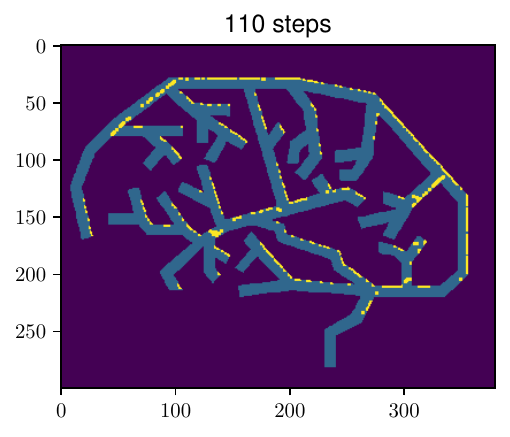}
	\includegraphics[width=\fwidth]{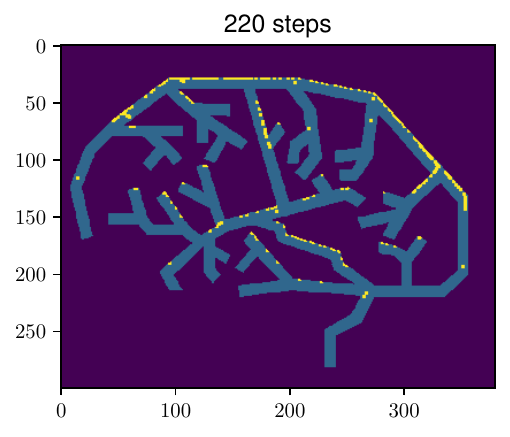}
	\includegraphics[width=\fwidth]{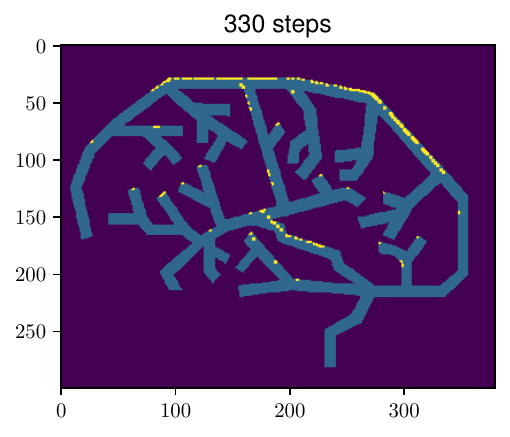}
	\includegraphics[width=\fwidth]{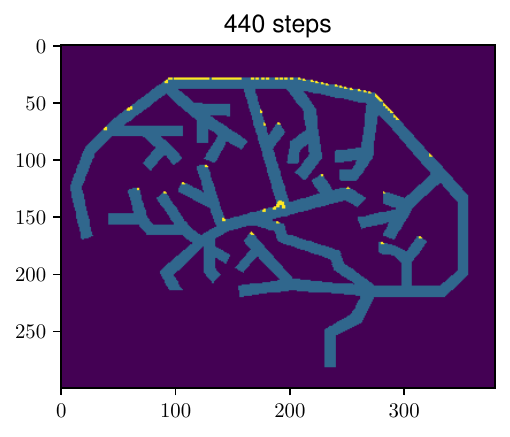}
	\includegraphics[width=\fwidth]{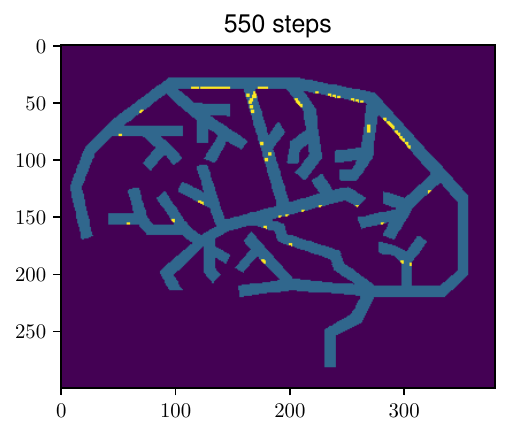}
	\includegraphics[width=\fwidth]{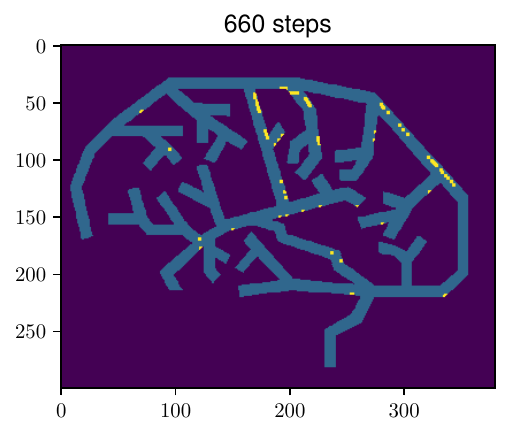}
	\includegraphics[width=\fwidth]{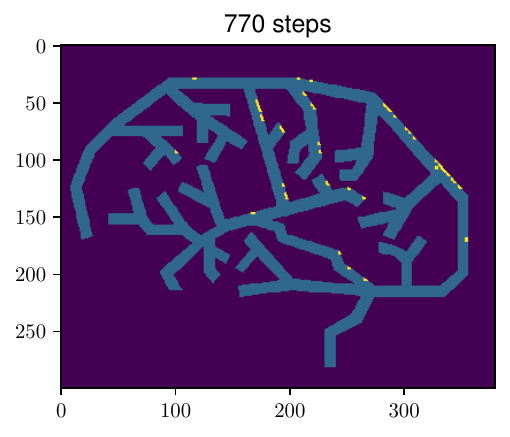}
	\includegraphics[width=\fwidth]{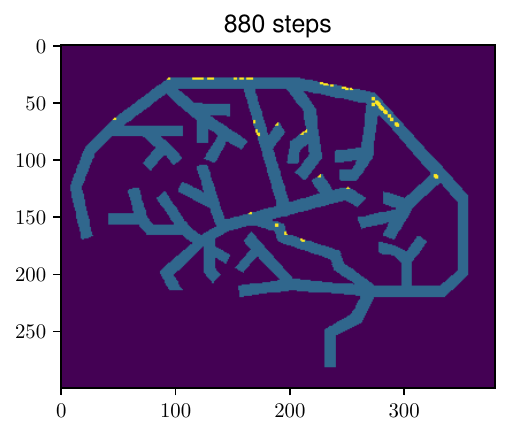}
	\includegraphics[width=\fwidth]{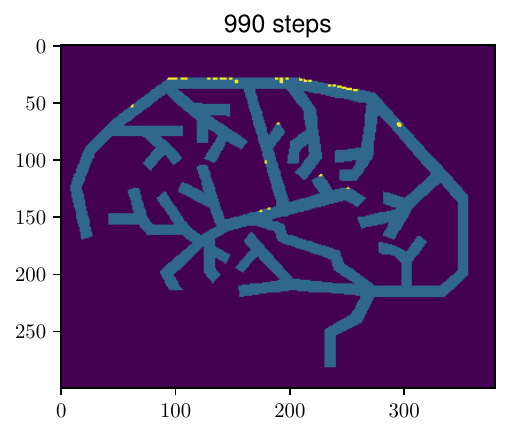}
	\includegraphics[width=\fwidth]{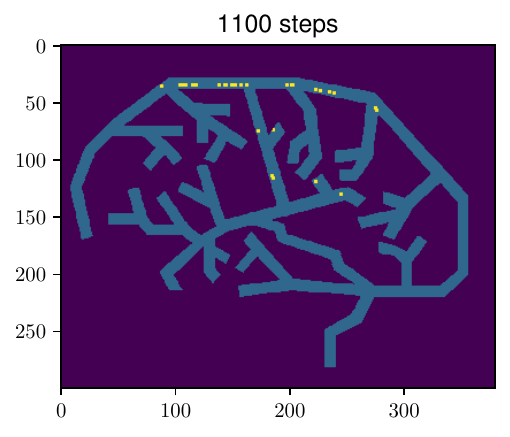}
	\includegraphics[width=\fwidth]{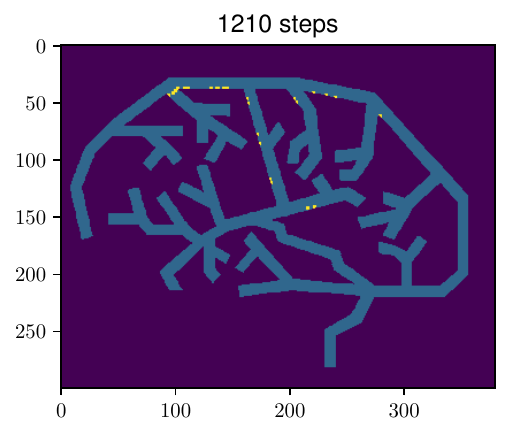}
	\includegraphics[width=\fwidth]{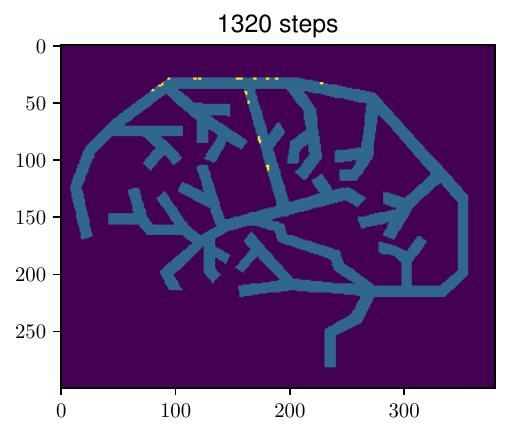}
	\includegraphics[width=\fwidth]{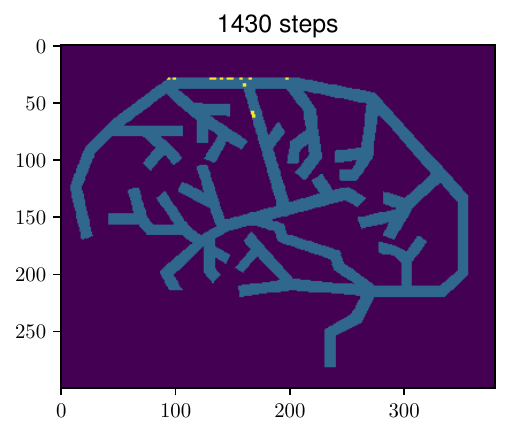}
	\includegraphics[width=\fwidth]{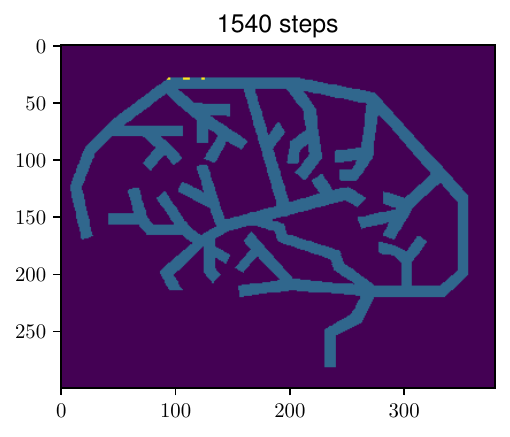}
	\caption{The process of gathering \num{1000} particles in \texttt{Brain} with \textsc{RL}. The visualization uses a maximum filter to improve visibility. We see how the algorithm quickly moves the particles out of the branches, but sometimes a particle drops back in.}\label{fig:tiltgather:over_time_graphic2}
  \end{figure*}

\section{Conclusions}
We have described a spectrum of methodological progress on an important problem of great practical relevance, in a discrete two-dimensional grid setting.

In particular, we showed that deciding whether ``short'' gathering sequences exist is \NP-complete, even for thin polyominoes.
In fact, we obtained two even stronger results, that is, the problem remains \NP-complete when target locations for the gathering sequence are prescribed, as well as when every pixel initially contains a particle. 

On the algorithmic side, we described an approach that drastically reduces the total number of particles (compared to the number of convex corners of the workspace), and different strategies that eventually merge two marked particles.
Furthermore, we obtained general upper bounds on the length of gathering sequences in both simple polyominoes and polyominoes with holes.

Our algorithmic simulations indicate the strength of our methods.
However, the different outcomes for deterministic as well as ML approaches
indicate that further, more detailed algorithmic studies are warranted to
understand the most successful line of attack; this includes studies 
of the necessary trade-off between the computation time and the number of actuation steps, but also includes modified models in which an actuation step may be able to move particles by more than an elementary distance.

\smallskip
Another interesting question is how we can deal with random errors in actuation and navigation? Our insights into
oblivious methods clearly indicate that these should remain tractable,
but more detailed considerations for frequency and amount of 
errors should provide quantifications and error-correcting approaches.

\smallskip
Finally, it is typically not necessary for our application scenarios 
to gather \emph{all} particles in a target area; moving an appropriate fraction
should usually suffice. 
\Cref{fig:tiltgather:over_time_graphic} visualizes a slightly different aspect, but still highlights the prospect that a considerably reduced number of actuation steps may be~achieved.

\smallskip
Although our exposition focuses on two-dimensional scenarios, a generalization to three-dimensional environments seems to be relatively straightforward in most cases, with some adjustments in runtime.

    \bibliography{bibliography.bib}

\begin{thebibliography}{10}

\bibitem{alpernbook}
Steve Alpern and Shmuel Gal.
\newblock {\em The theory of search games and rendezvous}.
\newblock Kluwer Academic Publishers, 2003.
\newblock \href {https://doi.org/10.1007/b100809} {\path{doi:10.1007/b100809}}.

\bibitem{anderson2001two}
Edward~J. Anderson and S{\'a}ndor~P. Fekete.
\newblock Two dimensional rendezvous search.
\newblock {\em Operations Research}, 49(1):107--118, 2001.
\newblock \href {https://doi.org/10.1287/opre.49.1.107.11191}
  {\path{doi:10.1287/opre.49.1.107.11191}}.

\bibitem{full_tilt}
Jose Balanza{-}Martinez, Austin Luchsinger, David Caballero, Rene Reyes,
  Angel~A. Cantu, Robert~T. Schweller, Luis~Angel Garcia, and Tim Wylie.
\newblock Full tilt: Universal constructors for general shapes with uniform
  external forces.
\newblock In {\em Symposium on Discrete Algorithms (SODA)}, pages 2689--2708,
  2019.
\newblock \href {https://doi.org/10.1137/1.9781611975482.167}
  {\path{doi:10.1137/1.9781611975482.167}}.

\bibitem{reconfiguring_swarm}
Aaron~T. Becker, Erik~D. Demaine, S{\'{a}}ndor~P. Fekete, Golnaz Habibi, and
  James McLurkin.
\newblock Reconfiguring massive particle swarms with limited, global control.
\newblock In {\em Symposium on Algorithms and Experiments for Sensor Systems,
  Wireless Networks and Distributed Robotics (ALGOSENSORS)}, pages 51--66,
  2013.
\newblock \href {https://doi.org/10.1007/978-3-642-45346-5\_5}
  {\path{doi:10.1007/978-3-642-45346-5\_5}}.

\bibitem{computation3_swarms}
Aaron~T. Becker, Erik~D. Demaine, S{\'{a}}ndor~P. Fekete, Jarrett Lonsford, and
  Rose Morris{-}Wright.
\newblock Particle computation: Complexity, algorithms, and logic.
\newblock {\em Natural Computing}, 18(1):181--201, 2019.
\newblock \href {https://doi.org/10.1007/s11047-017-9666-6}
  {\path{doi:10.1007/s11047-017-9666-6}}.

\bibitem{computation_swarms}
Aaron~T. Becker, Erik~D. Demaine, S{\'{a}}ndor~P. Fekete, and James McLurkin.
\newblock Particle computation: Designing worlds to control robot swarms with
  only global signals.
\newblock In {\em International Conference on Robotics and Automation (ICRA)},
  pages 6751--6756, 2014.
\newblock \href {https://doi.org/10.1109/ICRA.2014.6907856}
  {\path{doi:10.1109/ICRA.2014.6907856}}.

\bibitem{bmd+-pcdfbm-15}
Aaron~T. Becker, Erik~D. Demaine, S{\'a}ndor~P. Fekete, Hamed~Mohtasham Shad,
  and Rose Morris-Wright.
\newblock Tilt: {T}he video - {D}esigning worlds to control robot swarms with
  only global signals.
\newblock In {\em Symposium on Computational Geometry (SoCG)}, pages 16--18,
  2015.
\newblock \href {https://doi.org/10.4230/LIPIcs.SOCG.2015.16}
  {\path{doi:10.4230/LIPIcs.SOCG.2015.16}}.

\bibitem{BeckerFHKKKR020}
Aaron~T. Becker, S{\'{a}}ndor~P. Fekete, Li~Huang, Phillip Keldenich, Linda
  Kleist, Dominik Krupke, Christian Rieck, and Arne Schmidt.
\newblock Targeted drug delivery: Algorithmic methods for collecting a swarm of
  particles with uniform, external forces.
\newblock In {\em International Conference on Robotics and Automation (ICRA)},
  pages 2508--2514, 2020.
\newblock \href {https://doi.org/10.1109/ICRA40945.2020.9196551}
  {\path{doi:10.1109/ICRA40945.2020.9196551}}.

\bibitem{Chowdhury2015}
Sagar Chowdhury, Wuming Jing, and David~J. Cappelleri.
\newblock Controlling multiple microrobots: {R}ecent progress and future
  challenges.
\newblock {\em Journal of Micro-Bio Robotics}, 10(1--4):1--11, 2015.
\newblock \href {https://doi.org/10.1007/s12213-015-0083-6}
  {\path{doi:10.1007/s12213-015-0083-6}}.

\bibitem{DudekRW98}
Gregory Dudek, Kathleen Romanik, and Sue Whitesides.
\newblock Localizing a robot with minimum travel.
\newblock {\em {SIAM} Journal on Computing}, 27(2):583--604, 1998.
\newblock \href {https://doi.org/10.1137/S0097539794279201}
  {\path{doi:10.1137/S0097539794279201}}.

\bibitem{flocchini2019distributed}
Paola Flocchini, Giuseppe Prencipe, and Nicola Santoro.
\newblock {\em Distributed computing by mobile entities}.
\newblock Springer, 2019.
\newblock \href {https://doi.org/10.1007/978-3-030-11072-7}
  {\path{doi:10.1007/978-3-030-11072-7}}.

\bibitem{graesser2019foundations}
Laura Graesser and Wah~Loon Keng.
\newblock {\em Foundations of deep reinforcement learning: Theory and practice
  in Python}.
\newblock Addison-Wesley Professional, 2019.

\bibitem{stable-baselines}
Ashley Hill, Antonin Raffin, Maximilian Ernestus, Adam Gleave, Anssi
  Kanervisto, Rene Traore, Prafulla Dhariwal, Christopher Hesse, Oleg Klimov,
  Alex Nichol, Matthias Plappert, Alec Radford, John Schulman, Szymon Sidor,
  and Yuhuai Wu.
\newblock Stable baselines.
\newblock \url{https://github.com/hill-a/stable-baselines}, 2018.

\bibitem{makonitzny}
Matthias Konitzny.
\newblock Reinforcement learning for navigating particle swarms by global
  force.
\newblock Master's thesis, TU Braunschweig, Institute of Operating Systems and
  Computer Networks, Algorithms Division, 2019.
\newblock URL: \url{https://github.com/NeoExtended/baselines-lab}.

\bibitem{KonitznyLLFB22}
Matthias Konitzny, Yitong Lu, Julien Leclerc, S{\'{a}}ndor~P. Fekete, and
  Aaron~T. Becker.
\newblock Gathering physical particles with a global magnetic field using
  reinforcement learning.
\newblock In {\em International Conference on Intelligent Robots and Systems
  (IROS)}, pages 10126--10132, 2022.
\newblock \href {https://doi.org/10.1109/IROS47612.2022.9982256}
  {\path{doi:10.1109/IROS47612.2022.9982256}}.

\bibitem{litvinov2012high}
Julia Litvinov, Azeem Nasrullah, Timothy Sherlock, Yi-Ju Wang, Paul Ruchhoeft,
  and Richard~C. Willson.
\newblock High-throughput top-down fabrication of uniform magnetic particles.
\newblock {\em PloS One}, 7(5):e37440, 2012.
\newblock \href {https://doi.org/10.1371/journal.pone.0037440}
  {\path{doi:10.1371/journal.pone.0037440}}.

\bibitem{gathering_swarm}
Arun~V. Mahadev, Dominik Krupke, Jan{-}Marc Reinhardt, S{\'{a}}ndor~P. Fekete,
  and Aaron~T. Becker.
\newblock Collecting a swarm in a grid environment using shared, global inputs.
\newblock In {\em International Conference on Automation Science and
  Engineering (CASE)}, pages 1231--1236, 2016.
\newblock \href {https://doi.org/10.1109/COASE.2016.7743547}
  {\path{doi:10.1109/COASE.2016.7743547}}.

\bibitem{mathieu2007magnetic}
Jean-Baptiste Mathieu and Sylvain Martel.
\newblock Magnetic microparticle steering within the constraints of an {MRI}
  system: {P}roof of concept of a novel targeting approach.
\newblock {\em Biomedical Microdevices}, 9(6):801--808, 2007.
\newblock \href {https://doi.org/10.1007/s10544-007-9092-0}
  {\path{doi:10.1007/s10544-007-9092-0}}.

\bibitem{meghjani2012multi}
Malika Meghjani and Gregory Dudek.
\newblock Multi-robot exploration and rendezvous on graphs.
\newblock In {\em International Conference on Intelligent Robots and Systems
  (IROS)}, pages 5270--5276, 2012.
\newblock \href {https://doi.org/10.1109/IROS.2012.6386049}
  {\path{doi:10.1109/IROS.2012.6386049}}.

\bibitem{mellal2015magnetic}
Lyes Mellal, David Folio, Karim Belharet, and Antoine Ferreira.
\newblock Magnetic microbot design framework for antiangiogenic tumor therapy.
\newblock In {\em International Conference on Intelligent Robots and Systems
  (IROS)}, pages 1397--1402, 2015.
\newblock \href {https://doi.org/10.1109/IROS.2015.7353550}
  {\path{doi:10.1109/IROS.2015.7353550}}.

\bibitem{mnih2013playing}
Volodymyr Mnih, Koray Kavukcuoglu, David Silver, Alex Graves, Ioannis
  Antonoglou, Daan Wierstra, and Martin Riedmiller.
\newblock Playing {Atari} with deep reinforcement learning, 2013.
\newblock \href {https://arxiv.org/abs/1312.5602} {\path{arXiv:1312.5602}}.

\bibitem{mnih2015human}
Volodymyr Mnih, Koray Kavukcuoglu, David Silver, Andrei~A. Rusu, Joel Veness,
  Marc~G. Bellemare, Alex Graves, Martin~A. Riedmiller, Andreas Fidjeland,
  Georg Ostrovski, Stig Petersen, Charles Beattie, Amir Sadik, Ioannis
  Antonoglou, Helen King, Dharshan Kumaran, Daan Wierstra, Shane Legg, and
  Demis Hassabis.
\newblock Human-level control through deep reinforcement learning.
\newblock {\em Nature}, 518(7540):529--533, 2015.
\newblock \href {https://doi.org/10.1038/nature14236}
  {\path{doi:10.1038/nature14236}}.

\bibitem{pouponneau2009magnetic}
Pierre Pouponneau, Jean-Christophe Leroux, and Sylvain Martel.
\newblock Magnetic nanoparticles encapsulated into biodegradable microparticles
  steered with an upgraded magnetic resonance imaging system for tumor
  chemoembolization.
\newblock {\em Biomaterials}, 30(31):6327--6332, 2009.
\newblock \href {https://doi.org/10.1016/j.biomaterials.2009.08.005}
  {\path{doi:10.1016/j.biomaterials.2009.08.005}}.

\bibitem{schulman2017proximal}
John Schulman, Filip Wolski, Prafulla Dhariwal, Alec Radford, and Oleg Klimov.
\newblock Proximal policy optimization algorithms, 2017.
\newblock \href {https://arxiv.org/abs/1707.06347} {\path{arXiv:1707.06347}}.

\bibitem{williams1992simple}
Ronald~J. Williams.
\newblock Simple statistical gradient-following algorithms for connectionist
  reinforcement learning.
\newblock {\em Machine Learning}, 8:229--256, 1992.
\newblock \href {https://doi.org/10.1007/BF00992696}
  {\path{doi:10.1007/BF00992696}}.

\bibitem{zebrowski2007energy}
Pawel Zebrowski, Yaroslav Litus, and Richard~T. Vaughan.
\newblock Energy efficient robot rendezvous.
\newblock In {\em Computer and Robot Vision (CRV)}, pages 139--148, 2007.
\newblock \href {https://doi.org/10.1109/CRV.2007.27}
  {\path{doi:10.1109/CRV.2007.27}}.

\bibitem{zhang2017rearranging}
Yinan Zhang, Xiaolei Chen, Hang Qi, and Devin Balkcom.
\newblock Rearranging agents in a small space using global controls.
\newblock In {\em International Conference on Intelligent Robots and Systems
  (IROS)}, pages 3576--3582, 2017.
\newblock \href {https://doi.org/10.1109/IROS.2017.8206202}
  {\path{doi:10.1109/IROS.2017.8206202}}.

\bibitem{zhang2018assembling}
Yinan Zhang, Emily Whiting, and Devin Balkcom.
\newblock Assembling and disassembling planar structures with divisible and
  atomic components.
\newblock {\em Transactions on Automation Science and Engineering},
  15(3):945--954, 2018.
\newblock \href {https://doi.org/10.1007/978-3-030-43089-4_52}
  {\path{doi:10.1007/978-3-030-43089-4_52}}.

\end{thebibliography}
\end{document}